\documentclass[conference]{IEEEtran}
\usepackage[explicit]{titlesec}
\usepackage{graphicx}
\usepackage{verbatim}
\usepackage{epstopdf}
\usepackage{subfigure} 
\usepackage{algorithmic}
\usepackage{amsmath}
\usepackage{amsthm}
\usepackage{mathrsfs}
\usepackage{extarrows}
\usepackage{cite}
\newtheorem{theorem}{Theorem}

\usepackage{bbm}
\usepackage[bookmarks=false]{hyperref}
\usepackage{amssymb}

\newcommand{\RNum}[1]{\uppercase\expandafter{\romannumeral #1\relax}}
\newtheorem{lemma}{Lemma}
\newtheorem{proposition}{Proposition}

\usepackage{stfloats}

\usepackage{color, soul}
\usepackage{tikz,xcolor,hyperref}

\definecolor{lime}{HTML}{A6CE39}
\usepackage[left=0.68in,right=0.68in,top=0.67in,bottom=0.99in]{geometry}

\titlespacing{\section}{0pt}{1.2ex plus .0ex minus .0ex}{.3ex plus .0ex}
\titlespacing{\subsection}{0pt}{1.2ex plus .0ex minus .0ex}{.3ex plus .0ex}

\makeatletter
\DeclareRobustCommand{\orcidicon}{%
	\begin{tikzpicture}
	\draw[lime, fill=lime] (0,0) 
	circle [radius=0.16] 
	node[white] {{\fontfamily{qag}\selectfont \tiny ID}};    \draw[white, fill=white] (-0.0625,0.095) 
	circle [radius=0.007];    \end{tikzpicture}
	\hspace{-2mm}}
\foreach \x in {A, ..., Z}{%
	\expandafter\xdef\csname orcid\x\endcsname{\noexpand\href{https://orcid.org/\csname orcidauthor\x\endcsname}{\noexpand\orcidicon}}
}

\newcommand*\bigcdot{\mathpalette\bigcdot@{.5}}
\newcommand*\bigcdot@[2]{\mathbin{\vcenter{\hbox{\scalebox{#2}{$\m@th#1\bullet$}}}}}
\makeatother

\usepackage[linesnumbered,ruled,vlined]{algorithm2e}

\begin{document}
\title{Goal-oriented Tensor: Beyond AoI Towards Semantics-Empowered Goal-Oriented Communications}
\author{Aimin Li, 
	 Shaohua Wu, 
     and Sumei Sun \\
%
%
%
}

\maketitle
\allowdisplaybreaks
\begin{abstract}
The intricate interplay of source dynamics, unreliable channels, and \emph{staleness} of information has long been recognized as a significant impediment for the receiver to achieve accurate, timely, and most importantly, goal-oriented decision making. Thus, a plethora of promising metrics, such as \emph{Age of Information}, \emph{Value of Information}, and \emph{Mean Square Error}, have emerged to quantify these underlying adverse factors. Following this avenue, optimizing these metrics has indirectly improved the \emph{utility} of goal-oriented decision making. Nevertheless, no metric has hitherto been expressly devised to directly evaluate the \emph{utility} of a goal-oriented decision-making process. To this end, this paper investigates a novel performance metric, the Goal-oriented Tensor (GoT), to directly quantify the impact of semantic mismatches on goal-oriented decision making. Based on the GoT, we consider a \emph{sampler-decision maker} pair that work collaboratively and distributively to achieve a shared goal of communications. We formulate an infinite-horizon Decentralized Partially Observable Markov Decision Process (Dec-POMDP) to conjointly deduce the optimal deterministic sampling policy and decision-making policy. The simulation results reveal that the \emph{sampler-decision maker} co-design surpasses the current literature on AoI and its variants in terms of both goal achievement \emph{utility} and sparse sampling rate, signifying a notable accomplishment for a sparse sampler and goal-oriented decision maker co-design.
\end{abstract}

\IEEEpeerreviewmaketitle

\section{Introduction}\label{sectionI}
 {Age of Information} (AoI), a metric proposed in \cite{kaul2011minimizing}, has emerged as an important metric to capture the data \emph{freshness} perceived by the receiver. Since its inception, AoI has garnered significant research attention and has been extensively analyzed and optimized to improve the performance of queuing systems, physical-layer communications, MAC-layer communications, Internet of Things, etc. \cite{yatesAgeInformationIntroduction2021}. These research efforts are driven by the consensus that a freshly received message typically contains critical and valuable information, thereby improving the precision and timeliness of decision-making processes.

 Though AoI has been proven efficient in many freshness-critical applications, it exhibits several critical shortcomings. Specifically, ($a$) AoI fails to provide a direct measure of information value; 
 ($b$) AoI does not consider the content dynamics of source data; ($c$) AoI ignores the effect of End-to-End (E2E) information mismatch on the decision-making process. To address these limitations, numerous AoI variants have been extensively investigated. One typical approach in this research avenue is to impose a non-linear penalty on AoI \cite{kosta2017age,kosta2020cost,sun2019sampling} to evaluate the E2E ``\emph{dissatisfaction}'' degree resulted by \emph{stale} information. This non-linear penalty is called \emph{Value of Information} (VoI), which assists mitigate the shortcoming ($a$) discussed above. Other research attempt to address the shortcoming ($b$) \cite{AoCI,zhong2018two}. In \cite{AoCI}, \emph{Age of Changed Information} (AoCI) is proposed to address the ignorance of content dynamics of AoI. In this regard, unchanged statuses do not necessarily provide new information and thus are not prioritized for transmission. In \cite{zhong2018two}, the authors propose a novel age penalty named \emph{Age of Synchronization} (AoS), which represents the duration that has elapsed subsequent to the most recent synchronization. Mean Square Error (MSE) and its variants are introduced to address the shortcoming ($c$) \cite{zheng2020urgency, AoII}. In \cite{zheng2020urgency}, the authors introduce the context-aware weighting coefficient to propose the \emph{Urgency of Information} (UoI), which can measure the weighted MSE under contexts with varying levels of urgency. Moreover, considering that an E2E mismatch may exert a detrimental effect on the overall system's performance over time, the authors of \cite{AoII} propose a novel metric called \emph{Age of Incorrect Information} (AoII) to quantify the negative impact resulting from the duration of the E2E mismatch. The AoII metric reveals that both the degree and the duration of E2E semantics mismatch result in the \emph{utility} reduction to the subsequent decision making.
 
 Notwithstanding the above advancements, the question on \emph{how the E2E mismatch affects the utility of decision making has yet to be addressed}. \cite{Kountouris2020SemanticsEmpoweredCF,9551200,Salimnejad2023RealtimeRO,fountoulakisGoalorientedPoliciesCost2023} introduce a metric termed \emph{Cost of Error Actuation} to delve deeper into the cost resulting from the error actuation due to imprecise real-time estimations. Specifically, the \emph{Cost of Error Actuation} is denoted by an asymmetric zero diagonal matrix $\mathbf{C}$, with each value $C_{X_t,\hat{X}_t}$ representing the instant cost under the E2E mismatch status $(X_t,\hat{X}_t)_{X_t\ne\hat{X}_t}$. In this regard, the authors unveil that the \emph{utility} of decision making bears a close relation to the E2E semantic mismatch category, as opposed to the mismatch duration (AoII) or mismatch duration (MSE). For example, an E2E semantic mismatch category that a fire is estimated as no fire will result in higher cost; while the opposite scenario will result in lower cost.

Nonetheless, we notice that $i$) the method to obtain a \emph{Cost of Error Actuation} remains unclear, which implicitly necessitates a pre-established decision-making policy; $ii$) \emph{Cost of Error Actuation} does not consider the context-varying factors, which may also affect the decision-making \emph{utility}; $iii$) the zero diagonal property of the matrix implies the supposition that if $X_t=\hat{X}_t$, then $C_{X_t,\hat{X}_t}=0$, thereby signifying that errorless actuation necessitates no energy expenditure. To address these issues, the present authors have recently proposed a new metric referred to as GoT in \cite{li2023goaloriented}, which, compared to \emph{Cost of Error Actuation}, introduces new dimensions of the context $\Phi_{t}$ and the decision-making policy $\pi_A$ to describe the true \emph{utility} of decision making. In this paper, we further technically exploit the potential of GoT, the primary ingredients are as follows:

\noindent $\bullet$ We focus on the decision \emph{utility} issue by employing the GoT. A controlled Markov source is observed, wherein the transition of the source is dependent on both the decision making at the receiver and the contextual situation it is situated. In this case, the decision making will lead to three aspects in \emph{utility}: $i$) the future evolution of the source; $ii$) the instant cost at the source; $iii$) the energy and resources consumed by actuation. 

\noindent $\bullet$ We accomplish the goal-oriented \emph{sampler-decision maker} co-design, which, to the best of our knowledge, represents the first work that addresses the trade-off between sampling and decision making. Specifically, we formulate this problem as a two-agent infinite-horizon Dec-POMDP problem, with one agent embodying the sampler and the other representing the decision maker. Note that the optimal solution of even a finite-Horizon Dec-POMDP is known to be NEXP-complete \cite{bernstein2002complexity}, we design a RVI-Brute-Force-Search Algorithm to acquire the optimal joint sampling and decision making policies.

\section{System Model}\label{section III}

We consider a time-slotted perception-actuation loop where both the perceived semantics $X_t\in\mathcal{S}=\left\{s_1,\cdots,s_{|\mathcal{S}|}\right\}$ and context $\Phi_t\in\mathcal{V}=\left\{v_1,\cdots,v_{|\mathcal{V}|}\right\}$ are input into a semantic sampler, tasked with determining the significance of the present status $X_t$ and subsequently deciding if it warrants transmission via an unreliable channel. The semantics and context are extracted and assumed to perfectly describe the status of the observed process. The binary indicator, $a_S(t)=\pi_S(X_t,\Phi_{t},\hat{X}_t) \in \left\{0, 1\right\}$, signifies the sampling and transmission action at time slot $t$, with the value $1$ representing the execution of sampling and transmission, and the value $0$ indicating the idleness of the sampler. $\pi_S$ here represents the sampling policy. We consider a perfect and delay-free feedback channel \cite{Kountouris2020SemanticsEmpoweredCF,9551200,Salimnejad2023RealtimeRO,fountoulakisGoalorientedPoliciesCost2023}, with ACK representing a successful transmission and NACK representing the otherwise. The decision maker at the receiver will make decisions $a_A(t)\in\mathcal{A}_A=\left\{a_1,\cdots,a_{|\mathcal{A}_A|}\right\}$ base on the estimate $\hat{X}_t$, which will ultimately affect the \emph{utility} of the system. An illustration of our considered model is shown in Fig. \ref{systemmodel}. 

\begin{figure}[htbp]
	\centering
	\includegraphics[angle=0,width=0.45\textwidth]{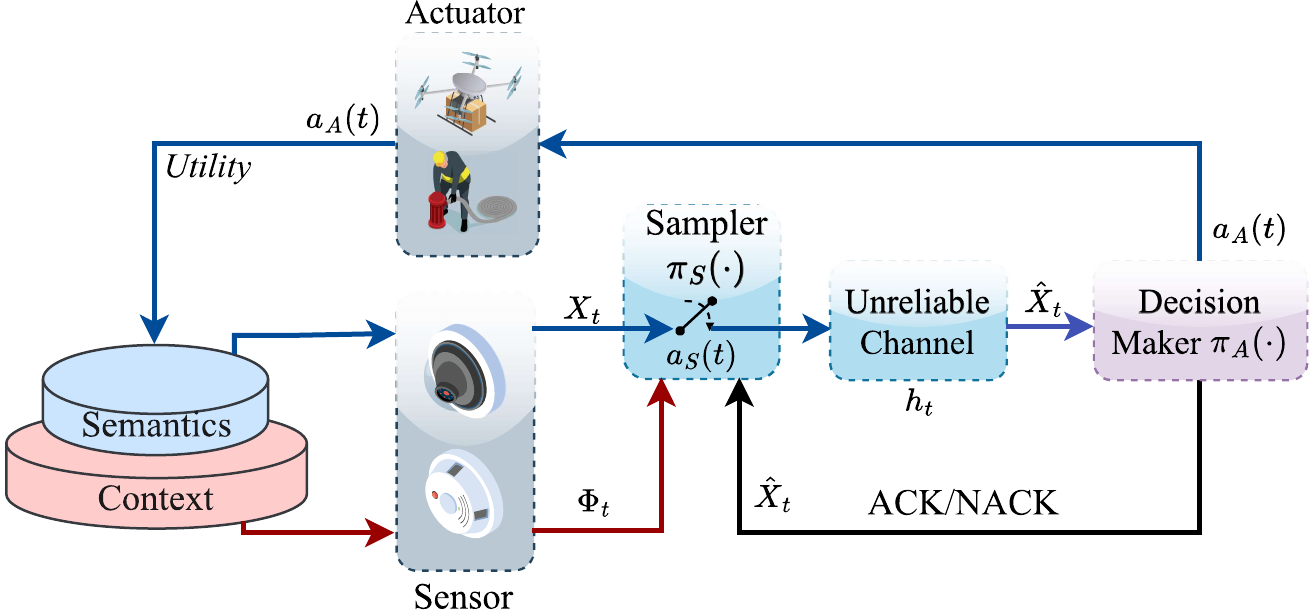}
	\caption{Illustration of our considered system where transmitted semantic status arrives at a receiver for decision making to achieve a certain goal.}\label{systemmodel}
\end{figure}

\subsection{Semantics and Context Dynamics}

We consider a controlled Discrete Markov source:
\begin{equation}\label{Source}
\Pr \left( {X_{t+1}=s_u\left| {X_t=s_i, a_A(t)=a_m, \Phi_t=v_k} \right.} \right)=p_{i,u}^{(k,m)}.
\end{equation}
Here the dynamics of the source is dependent on both the decision making $a_A(t)$ and context $\Phi_{t}$. Furthermore, we take into account the variations in context $\Phi_t$, characterized by the transition probability:
\begin{equation}\label{context}
\Pr \left( {\Phi_t=v_r\left| {\Phi_{t+1}=v_k} \right.} \right)=p_{k,r}.
\end{equation}
In general, the dynamics of semantics and context are independent of each other.

\subsection{Unreliable Channel and Estimate Transition}
We assume that the channel realizations exhibit independence and identical distribution (i.i.d.) across time slots, following a Bernoulli distribution. Particularly, the channel realization $h_t$ assumes a value of $1$ in the event of successful transmission, and $0$ otherwise. Accordingly, we define the probability of successful transmission as $\Pr\left(h_t=1\right)=p_S$ and the failure probability as $\Pr\left(h_t=0\right)=1-p_S$. To characterize the dynamic process of $\hat{X}_t$, we consider two cases as described below:

\noindent $\bullet$ $a_S(t)=0$. In this case, the sampler and transmitter remain idle, manifesting that there is no new knowledge given to the receiver, \emph{i.e.}, $\hat{X}_{t+1}=\hat{X}_{t}$. As such, we have:
\begin{equation}\label{as0}
\Pr\left({\hat{X}_{t+1}=x\left|\hat{X}_t=s_j,a_S(t)=0\right.}\right)=\mathbbm{1}_{\left\{x=s_j\right\}}.
\end{equation}

\noindent $\bullet$ $a_S(t)=1$. In this case, the sampler and transmitter transmit the current semantic status $X_t$ through an unreliable channel. As the channel is unreliable, we differentiate between two distinct situations: $h_t=1$ and $h_t=0$:

\noindent (a) $h_t=1$. In this case, the transmission is successful. As such, the estimate at the receiver $\hat{X}_{t+1}$ is nothing but $X(t)$, and the transition probability is
\begin{equation}
\begin{aligned}
&\Pr\left({\hat{X}_{t+1}=x\left|\hat{X}_t=s_j,X_t=s_i,a_S(t)=1,h_t=1\right.}\right)\\
&=\mathbbm{1}_{\left\{x=s_i\right\}}.
\end{aligned}
\end{equation}

\noindent (b) $h_t=0$. In this case, the transmission is not successfully decoded by the receiver. As such, the estimate at the receiver $\hat{X}_{t+1}$ remains $\hat{X}(t)$. In this way, the transition probability is
\begin{equation}
\small
\begin{aligned}
&\Pr\left({\hat{X}_{t+1}=x\left|\hat{X}_t=s_j,X_t=s_i,a_S(t)=1,h_t=0\right.}\right)=\mathbbm{1}_{\left\{x=s_j\right\}}.
\end{aligned}
\end{equation}
As the channel realization $h_t$ is independent with the process of $X_t$, $\hat{X}_t$, and $a_S(t)$, we have that
\begin{equation}\label{as1}
\small
\begin{aligned}
&\Pr\left({\hat{X}_{t+1}=x\left|\hat{X}_t=s_j,X_t=s_i,a_S(t)=1\right.}\right)\\
&=\sum_{h_t}p\left(h_t\right)\Pr\left({\hat{X}_{t+1}=x\left|\hat{X}_t=s_j,X_t=s_i,a_S(t)=1,h_t\right.}\right)\\
&=p_S\cdot\mathbbm{1}_{\left\{x=s_i\right\}}+(1-p_S)\cdot\mathbbm{1}_{\left\{x=s_j\right\}}.
\end{aligned}
\end{equation}
Combing (\ref{as0}) with (\ref{as1}) yields the dynamics of the estimate.
\subsection{Goal-oriented Decision Making and Actuating}
We note that the previous works primarily focus on minimizing the open-loop freshness-related or error-related penalty for a transmitter-receiver system. Nevertheless, irrespective of the \emph{fresh} delivery or accurate end-to-end timely reconstruction, the ultimate goal of such optimization efforts is to ensure precise and effective decision-making. To this end, we broaden the open-loop transmitter-receiver information flow to include a perception-actuation closed-loop \emph{utility} flow by incorporating the decision-making and actuation processes. As a result, decision-making and actuation enable the conversion of status updates into ultimate effectiveness. Here the decision making at time slot $t$ follows that $a_A(t)=\pi_A(\hat{X}_t)$, with $\pi_A$ representing the deterministic decision-making policy.

\subsection{Metric: Goal Characterization Through GoT}\label{section II}
A three-dimension GoT could be defined by a mapping \footnote{It is important to note that the GoT could be expanded into higher dimensions by integrating additional components, including actuation policies, task-specific attributes, and other pertinent factors.}:
\begin{equation}
(X_t,\Phi_t,\hat{X}_t)\in \mathcal{S}\times \mathcal{V}\times\mathcal{S}\overset{\mathcal{L}}{\rightarrow} \mathrm{GoT}(t)\in\mathbb{R}.
\end{equation}
In this regard, the GoT, denoted by $\mathcal{L}(X_t,\Phi_t,\hat{X}_t)$ or $\mathrm{GoT}(t)$, indicates the instant cost of the system at time slot $t$, with the knowledge of $(X_t,\Phi_t,\hat{X}_t)$. From \cite{li2023goaloriented}, we have shown that a GoT, given a specific triple-tuple $(X_t, \hat{X}_t,\Phi_t)$ and a decision-making policy $\pi_A$, could be calculated by 
\begin{equation}\label{got}
\begin{aligned}
&\mathrm{GoT}^{\pi_A}(t)=\mathcal{L}(X_t,\Phi_t,\hat{X}_t,\pi_A)\\
&=\left[C_1(X_t,\Phi_t) - C_2(\pi_A(\hat{X}_t))\right]^+
+ C_3(\pi_A(\hat{X}_t)),
\end{aligned}
\end{equation}
where the status inherent cost $C_1(X_t,\Phi_t)$ quantifies the inherent cost under different semantics-context pairs $(X_t, \Phi_t)$ in the absence of external influences; the actuation gain cost $C_2(\pi_A(\hat{X}_t))$ quantifies the prospective {reduction in severity} resulting from the actuation $\pi_A(\hat{X}(t))$; the actuation inherent cost $C_3(\pi_A(\hat{X}_t))$ reflects the resources consumed by a particular actuation $\pi_A(\hat{X}(t))$. The ramp function $\left[\cdot\right]^+$ ensures that any actuation $\pi_A(\hat{X}_t)$ reduces the cost to a maximum of 0. A visualization of a specific GoT construction is shown in Fig. \ref{gotconstructing}. The GoT in Fig. 2 is obtained through the following definition:
\begin{equation}\label{gotexample}
\begin{array}{ll}
C_1(X_t,\Phi_t)= \left(\begin{array}{c|ccc}
~ & 0 & 1 & 2 \\
\hline
0 & 0 & 1 & 3 \\
1 & 0 & 2 & 5
\end{array}\right),& \pi_A(\hat{X}_t)=\left[0,1,2\right], \\
C_2(\pi_A(\hat{X}_t))=2\pi_A(\hat{X}_t), &C_3(\pi_A(\hat{X}_t)) =\pi_A(\hat{X}_t).
\end{array}
\end{equation}

\begin{figure}[htbp]
	\centering
	\includegraphics[angle=0,width=0.45\textwidth]{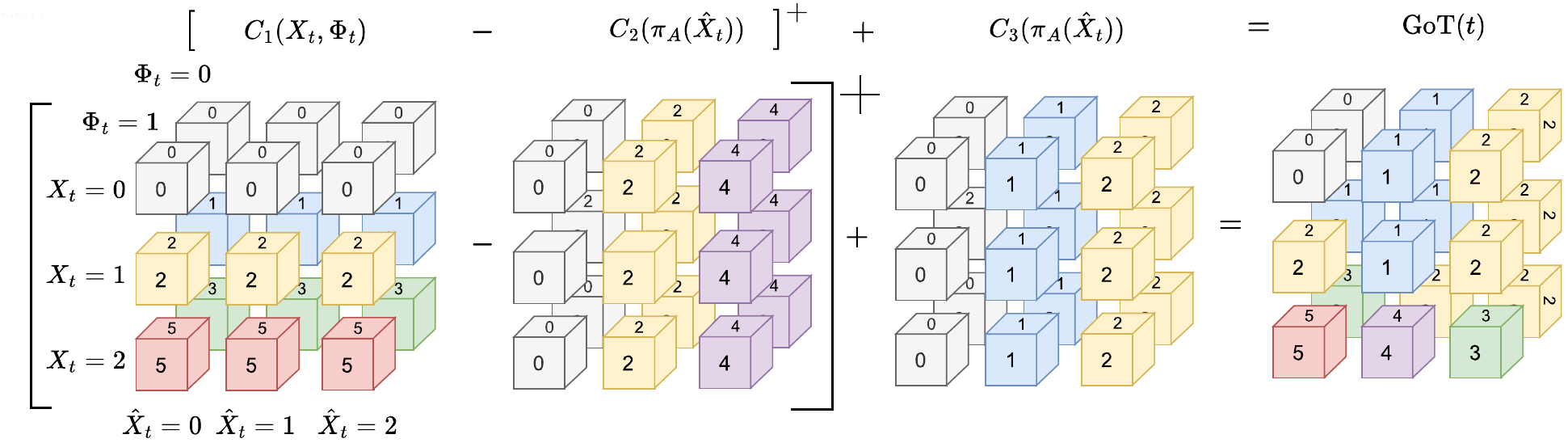}
	\caption{A visualized example for characterizing the GoT through (\ref{got}) and (\ref{gotexample}).}\label{gotconstructing}
\end{figure}
\begin{figure}[htbp]
	\centering
	\includegraphics[angle=0,width=0.45\textwidth]{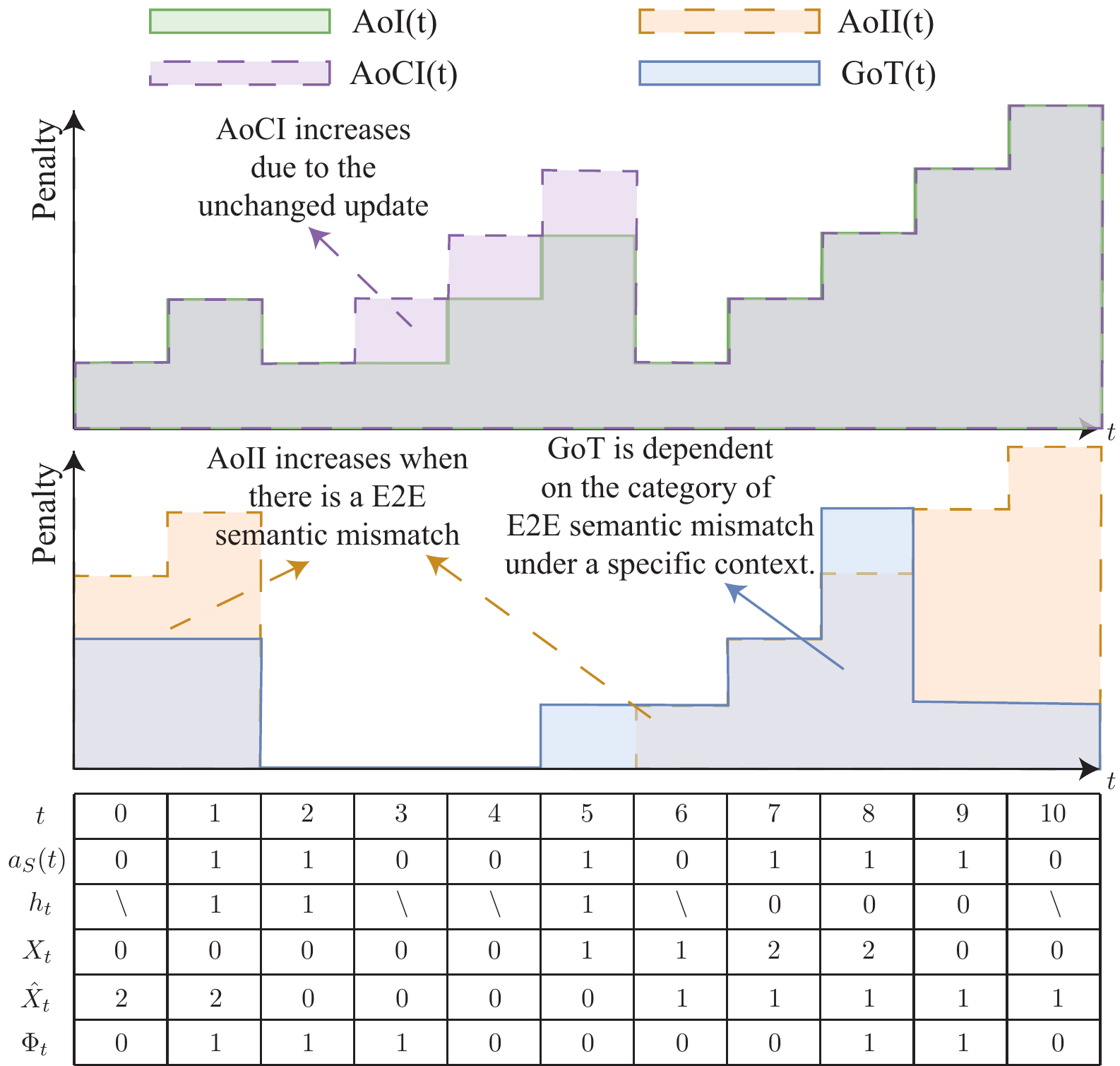}
	\caption{An illustration of AoI, AoCI, AoII, and GoT in a time-slotted status update system. Here, the value of GoT is obtained from the tensor obtained on the right-hand side of Fig. \ref{gotconstructing}.}\label{evolution}
\end{figure}

Fig. \ref{evolution} exhibits an instantaneous progression of AoI, AoCI, AoII, and GoT. From the time slots $t=0,1,9,10$ in Fig. \ref{evolution}, the inherent limitation of AoII emerges conspicuously, as a duration of mismatch may not necessarily culminate in a cost increase. Instead, the category of E2E semantic mismatch will make sense to the true instant cost.
\section{Problem Formulation and Solution}\label{sectionIV}
Conventionally, the formulation of sampling policy has been designed independently from the decision-making process. A typical illustration of this two-stage methodology involves first determining the optimal sampling policy based on AoI or its variants, and subsequently accomplishing effective decision making. This two-stage separate design arises from the inherent limitation of existing metrics that they fail to capture the closed-loop decision \emph{utility}. Nevertheless, the metric GoT empowers us to undertake a co-design of sampling and decision making. We explore the \textit{team decision theory}, wherein two agents, one embodying the sampler and the other the decision maker, collaborate to achieve a shared goal. We aim at determining a joint deterministic policy $\boldsymbol{\pi}_C=(\pi_S,\pi_A)$ that minimizes the long-term average cost of the system. It is considered that the sampling and transmission of an update also consumes energy, incurring a $C_s$ cost. In this case, the instant cost of the system could be clarified by $\mathrm{GoT}^{\pi_A}(t)+C_s\cdot a_S(t)$, and the problem is characterized as: 
\begin{equation}\label{p1}
\begin{array}{*{20}{c}}
{{{\cal P}} 1:}&{\mathop {\min }\limits_{{\boldsymbol{\pi}_C}  \in \Upsilon } \mathop {\lim \sup }\limits_{T \to \infty } \frac{1}{T}{\mathbb{E}^{{\boldsymbol{\pi}}_C} }\left( {\sum\limits_{t = 0}^{T - 1} {\mathrm{GoT}^{\pi_A}(t)+C_s\cdot a_S(t)} } \right)}\\
\end{array},
\end{equation}
where $\boldsymbol{\pi}_C=(\pi_S,\pi_A)$ denotes the joint sampling and decision policy, comprising $\pi_S=(a_S(0),a_S(1),\cdots)$ and $\pi_A=(a_A(0),a_A(1),\cdots)$, which correspond to the sampling action sequence and actuation sequence, respectively. Note that $\mathrm{GoT}^{\pi_A}(t)$ is characterized by (\ref{got}).
\subsection{Dec-POMDP Formulation}
To solve the problem $\mathcal{P}1$, we ought to formulate a
\textit{Decentralized Partially Observable Markov Decision Processes} (DEC-POMDP) problem, which is initially introduced in \cite{bernstein2002complexity} to solve the cooperative sequential decision issues for distributed multi-agents. Within a Dec-POMDP framework, a team of agents cooperates to achieve a shared goal, relying solely on their localized knowledge. A typical Dec-POMDP is denoted by a tuple $\mathscr{M}_{DEC-POMDP}\triangleq\left\langle n, \mathcal{I}, \mathcal{A}, \mathcal{T}, \Omega, \mathcal{O}, \mathcal{R}  \right\rangle$:

\noindent $\bullet$ $n$ denotes the number of agents. We have $n=2$ in the considered model, signifying the presence of two agents: one agent $\mathcal{A}gent_S$ embodies the sampler, while the other represents the decision maker, denoted by $\mathcal{A}gent_A$.

\noindent $\bullet$ $\mathcal{I}$ is the finite set of the global system status, characterized by $(X_t, \hat{X}_t,\Phi_t)\in \mathcal{S}\times\mathcal{S}\times\mathcal{V}$. For the sake of brevity, we henceforth denote $\mathbf{W}_t=(X_t, \hat{X}_t,\Phi_t)$ in the squeal.

\noindent $\bullet$ $\mathcal{T}$ is the transition function defined by
\begin{equation}
\mathcal{T}\left(\mathbf{w},\mathbf{a},\mathbf{w}'\right)\triangleq\Pr(\mathbf{W}_{t+1}=\mathbf{w}'|\mathbf{W}_t=\mathbf{w},\mathbf{a}_t=\mathbf{a}),
\end{equation}
which is defined by the transition probability from global status $\mathbf{W}_t=\mathbf{w}$ to status $\mathbf{W}_{t+1}=\mathbf{w}'$, after the agents in the system taking a joint action $\mathbf{a}_t=\mathbf{a}=(a_S(t),a_A(t))$. {For the sake of concise notation, we let $p(\mathbf{w}'|\mathbf{w},\mathbf{a})$ symbolize $\mathcal{T}\left(\mathbf{w},\mathbf{a},\mathbf{w}'\right)$ in the subsequent discourse.} Then, the transition functions can be calculated in lemma \ref{as11}
\begin{lemma} The transition functions of the Dec-POMDP: \label{as11}
	\begin{equation}
	\begin{aligned}
	&p\left((s_u,x,v_r)\left|(s_i,s_j,v_k),(1,a_m)\right.\right)=\\
	&p_{i,u}^{(k,m)}\cdot p_{k,r} \cdot \left(p_S\cdot\mathbbm{1}_{\left\{x=s_i\right\}}+(1-p_S)\cdot\mathbbm{1}_{\left\{x=s_j\right\}}\right)\\
	\end{aligned},
	\end{equation}
	\begin{equation}\label{as00}
	\begin{aligned}
	&p\left((s_u,x,v_r)\left|(s_i,s_j,v_k),(0,a_m)\right.\right)=\\
	&p_{i,u}^{(k,m)}\cdot p_{k,r} \cdot \mathbbm{1}_{\left\{x=s_j\right\}}
	\end{aligned},
	\end{equation}
	for any $x\in\mathcal{S}$ and indexes $i$, $j$, $u\in\left\{1,2,\cdots,|\mathcal{S}|\right\}$, $k$, $r\in\left\{1,2,\cdots,|\mathcal{V}|\right\}$, and $m\in\left\{1,2,\cdots,|\mathcal{A}_A|\right\}$.
\end{lemma}
\begin{proof}
	By taking into account the \textit{conditional independence} among $X_{t+1}$, $\Phi_{t+1}$, and $X_{t+1}$, given $(X_{t},\Phi_{t},X_{t})$ and $\mathbf{a}(t)$, the transition functions can be derived by incorporating the dynamics in equations (\ref{Source}), (\ref{context}), (\ref{as0}), and (\ref{as1}). 
\end{proof}

\noindent $\bullet$ $\mathcal{A}=\mathcal{A}_S\times\mathcal{A}_A$, with $\mathcal{A}_S\triangleq\left\{0,1\right\}$ representing the action set of the sampler, and $\mathcal{A}_A\triangleq\left\{a_0,\cdots,a_{M-1}\right\}$ representing the action set of the decision maker.

\noindent $\bullet$ $\Omega=\Omega_S\times \Omega_A$, with $\Omega_S$ signifies the sampler's observation domain. In this instance, the sampler $\mathcal{A}gent_S$ is entirely observable, with $\Omega_S$ encompassing the comprehensive system state ${o}_S^{(t)}=\mathbf{W}_t$. $\Omega_A$ signifies the actuator's observation domain. In this case, the decision-maker $\mathcal{A}gent_A$ is partially observable, with $\Omega_A$ comprising $o_A^{(t)}=\hat{X}(t)$. The joint observation at time instant $t$ is denoted by $\mathbf{o}_t=(o_S^{(t)},o_A^{(t)})$.

\noindent $\bullet$ $\mathcal{O}=\mathcal{O}_S\times\mathcal{O}_A$ represents the observation function, where $\mathcal{O}_S$ and $\mathcal{O}_A$ denotes the observation function of the sampler $\mathcal{A}gent_S$ and the actuator $\mathcal{A}gent_A$, respectively, defined as:
\begin{equation}
\begin{aligned}
\mathcal{O}(\mathbf{w}, \mathbf{o})\triangleq\Pr(\mathbf{o}_t=\mathbf{o}|\mathbf{W}_t=\mathbf{w}),\\
\mathcal{O}_S(\mathbf{w}, o_S)\triangleq\Pr(o_S^{(t)}=o_S|\mathbf{W}_t=\mathbf{w}),\\
\mathcal{O}_A(\mathbf{w}, o_A)\triangleq\Pr(o_A^{(t)}=o_A|\mathbf{W}_t=\mathbf{w}).
\end{aligned}
\end{equation}	
The observation function of an agent $\mathcal{A}gent_i$ signifies the conditional probability of agent $\mathcal{A}gent_i$ perceiving $o_i$, contingent upon the prevailing global system state as $\mathbf{W}_t=\mathbf{w}$. For the sake of brevity, we henceforth let $p_A(o_A|\mathbf{w})$ represent $\mathcal{O}_A(\mathbf{w}, o_A)$ and $p_S(o_S|\mathbf{w})$ represent $\mathcal{O}_S(\mathbf{w}, o_A)$ in the subsequent discourse. In our considered model, the observation functions are deterministic, characterized by lemma \ref{obf}.
\begin{lemma} The observation functions of the Dec-POMDP: \label{obf}
	\begin{equation} 
	\begin{aligned}
	p_S\left((s_u,s_r,v_q)|(s_i,s_j,v_k)\right)&=\mathbbm{1}_{\left\{(s_u,s_r,v_q)=(s_i,s_j,v_k)\right\}},\\
	p_A\left(s_z|(s_i,s_j,v_k)\right)&=\mathbbm{1}_{\left\{s_z=s_j\right\}}.
	\end{aligned}
	\end{equation}
	for indexes $z$, $i$, $j$, $u$, $r\in\left\{1,2,\cdots\,|\mathcal{S}|\right\}$, and $k$, $q\in\left\{1,2,\cdots\,|\mathcal{V}|\right\}$.
\end{lemma}

\noindent $\bullet$ $\mathcal{R}$ is the reward function, characterized by a mapping $\mathcal{I}\times\mathcal{A}\rightarrow\mathbb{R}$. In the long-term average reward maximizing setup, resolving a Dec-POMDP is equivalent to addressing the following problem:
\begin{equation}
\begin{array}{*{20}{c}}
{\mathop {\min }\limits_{{\boldsymbol{\pi}_C}  \in \Upsilon } \mathop {\lim \sup }\limits_{T \to \infty } \frac{1}{T}{\mathbb{E}^{{\boldsymbol{\pi}}_C} }\left( -{\sum\limits_{t = 0}^{T - 1} r(t) } \right)}
\end{array}.
\end{equation}
Subsequently, to establish the congruence with the problem $\mathcal{P}1$, the reward function is correspondingly defined as:
\begin{equation}
r(t)=\mathcal{R}^{\pi_A}(\mathbf{w},a_S)=-\mathrm{GoT}^{\pi_A}(t)-C_s\cdot a_S(t).
\end{equation}

\subsection{Solutions to the Infinite-Horizon Dec-POMDP}

In general, solving a Dec-POMDP is known to be NEXP-complete for the finite-horizon setup \cite{bernstein2002complexity}. For an infinite-horizon Dec-POMDP problem, finding an optimal policy for a Dec-POMDP problem is known to be undecidable. Nevertheless, within our considered model, both the sampling and decision-making processes are considered to be deterministic, given as $a_S(t)=\pi_S(\mathbf{w})$ and $a_A(t)=\pi_A(o_A)$. In this case, it is feasible to determine a joint optimal policy via Brute Force Search across the decision-making policy space.

The idea is based on the finding that, given a deterministic decision policy $\pi_A$, the sampling problem can be formulated as a standard fully observed MDP problem denoted by $\mathscr{M}^{\pi_A}_{\mathrm{MDP}}\triangleq\langle\mathcal{I},\mathcal{T}^{\pi_A},\mathcal{A}_S,\mathcal{R}\rangle$. 

\begin{proposition}\label{proposition1}
	Given a deterministic decision-making policy $\pi_A$, the optimal sampling problem could be formulated by a typical fully observed MDP problem $\mathscr{M}^{\pi_A}_{\mathrm{MDP}}\triangleq\langle\mathcal{I},\mathcal{A}_S,\mathcal{T}_{\mathrm{MDP}}^{\pi_A},\mathcal{R}\rangle$, where the elements are given as follows:
	
	\noindent $\bullet$ $\mathcal{I}$: the same as the pre-defined Dec-POMDP tuple.
	 
	\noindent $\bullet$ $\mathcal{A}_S=\left\{0,1\right\}$: the sampling and transmission action set.
	
	\noindent $\bullet$ $\mathcal{T}^{\pi_A}$: the transition function given a deterministic decision policy $\pi_A$, which is 
	\begin{equation}
	\begin{aligned}
	&\mathcal{T}^{\pi_A}(\mathbf{w},a_S,\mathbf{w}')=p^{\pi_A}\left(\mathbf{w}'|\mathbf{w},a_S\right)\\
	&=\sum_{o_A\in\mathcal{O}_A}p\left(\mathbf{w}'|\mathbf{w},(a_S,\pi_A(o_A))\right)p_A(o_A|\mathbf{w})
	\end{aligned},
	\end{equation}
	where $p\left(\mathbf{w}'|\mathbf{w},(a_S,\pi_A(o_A))\right)$ could be obtained by Lemma \ref{as11} and $p(o_A|\mathbf{w})$ could be obtained by Lemma \ref{obf}.
	
	\noindent $\bullet$ $\mathcal{R}$: the same as the pre-defined Dec-POMDP tuple.
	
\end{proposition}

We now proceed to solve the MDP problem $\mathscr{M}^{\pi_A}_{\mathrm{MDP}}$. To deduce the optimal sampling policy under decision policy $\pi_A$, it is imperative to resolve the Bellman equations \cite{bertsekas2012dynamic}:
\begin{equation}
\begin{aligned}
&\theta^*_{\pi_A}+V_{\pi_A}(\mathbf{w})=\\
&\mathop{\max}\limits_{a_S\in\mathcal{A}_A}\left\{\mathcal{R}^{\pi_A}(\mathbf{w},a_S)+\sum_{\mathbf{w}'\in\mathcal{I}}p(\mathbf{w}'|\mathbf{w},a_S)V_{\pi_A}(\mathbf{w}')\right\},\\
\end{aligned}
\end{equation}
where $V^{\pi_A}(\mathbf{w})$ is the value function and $\theta^*_{\pi_A}$ is the optimal long-term average reward given the decision policy $\pi_A$. We apply the relative value iteration (RVI) algorithm to solve this problem. The details are shown in Algorithm \ref{Algorithm 1}:
\begin{algorithm}
	\caption{The RVI Algorithm to Solve the MDP Given the decision policy $\pi_A$}
	\label{Algorithm 1}
	\LinesNumbered
	\KwIn{The MDP tuple $\mathscr{M}^{\pi_A}_{\mathrm{MDP}}$, $\epsilon$, $\pi_A$;}
	Initialization: $\forall\mathbf{w}\in\mathcal{I}$, $\tilde{V}^0_{\pi_A}(\mathbf{w})=0$, $\tilde{V}^{-1}_{\pi_A}(\mathbf{w})=\infty$, $k=0$ \;
	Choose $\mathbf{w}^{ref}$ arbitrarily\;
		\While {$||\tilde{V}_{\pi_A}^k(\mathbf{w})-\tilde{V}_{\pi_A}^{k-1}(\mathbf{w})||\ge \epsilon$}
		{
			$k=k+1$\;
		\For{$\mathbf{w}\in\mathcal{I}-\mathbf{w}^{ref}$}
		{
			$\begin{scriptsize}
			\begin{aligned}
			&\tilde{V}_{\pi_A}^k(\mathbf{w})=-g_k+\\
			&\mathop{\max}\limits_{a_S}\left\{\mathcal{R}(\mathbf{w},a_S)+\sum_{\mathbf{w}'\in\mathcal{I}-\mathbf{w}^{ref}}p(\mathbf{w}'|\mathbf{w},a_S)\tilde{V}^{k-1}_{\pi_A}(\mathbf{w}')\right\};
			\end{aligned}
			\end{scriptsize}$
		}			
		}
	$\begin{aligned}
		&\theta^*(\pi_A,\pi_S^*)=-\tilde{V}^{k}_{\pi_A}(\mathbf{w})\\
		&\mathop{\max}\limits_{a_S\in\mathcal{A}_S}\left\{\mathcal{R}(\mathbf{w},a_S)+\sum_{\mathbf{w}'\in\mathcal{I}}p(\mathbf{w}'|\mathbf{w},a_S)\tilde{V}^{k}_{\pi_A}(\mathbf{w}')\right\}
	\end{aligned}$\;
	\For{$\mathbf{w}\in\mathcal{I}$}{
	$\pi_S^{*}(\pi_A,\mathbf{w})=\mathop{\arg\max}\limits_{a_S}\left\{\mathcal{R}(\mathbf{w},a_S)+\sum_{\mathbf{w}'\in\mathcal{I}}p(\mathbf{w}'|\mathbf{w},a_S)\tilde{V}^{k}_{\pi_A}(\mathbf{w}')\right\}$;}
	\KwOut{$\pi_S^{*}(\pi_A)$, $\theta^*(\pi_A,\pi_S^*)$}
\end{algorithm}

With Proposition \ref{proposition1} and Algorithm \ref{Algorithm 1} in hand, we could then perform a Brute Force Search across the decision policy space $\Upsilon_A$, thereby acquiring the joint sampling-decision-making policy. The algorithm, called RVI-Brute-Force-Search Algorithm, is elaborated in Algorithm \ref{Algorithm 2}.
\begin{algorithm}
	\caption{The RVI-Brute-Force-Search Algorithm}
	\label{Algorithm 2}
	\LinesNumbered
	\KwIn{The Dec-POMDP tuple $\mathscr{M}_{DEC-POMDP}$;}
	\For{$\pi_A\in\Upsilon$}{
		Formulate the MDP problem $\mathscr{M}^{\pi_A}_{\mathrm{MDP}}\triangleq\langle\mathcal{I},\mathcal{A}_S,\mathcal{T}_{\mathrm{MDP}}^{\pi_A},\mathcal{R}\rangle$ as given in Proposition \ref{proposition1}\;
		Run Algorithm 1 to obtain $\pi_S^*(\pi_A)$ and $\theta^*({\pi_A},\pi_S^*)$\;}
	Calculate the optimal joint policy:
	$\begin{cases}
	&\pi_A^* =\mathop{\arg\min}\nolimits_{\pi_A}\theta_{\pi_A}^* \\
	&\pi_S^*=\pi_S(\pi^*_A)
	\end{cases}$\;
	\KwOut{$\pi_S^{*}$, $\pi_A^*$}
\end{algorithm}
\begin{theorem}
	The RVI-Brute-Force-Search Algorithm (Algorithm \ref{Algorithm 2}) could achieve the optimal joint deterministic policies $(\pi_S^*,\pi_A^*)$, given that the transition function $\mathcal{T}^{\pi_A}$ follows a unichan. 
	
\end{theorem}
\begin{proof}\renewcommand{\qedsymbol}{}
	If the the transition function $\mathcal{T}^{\pi_A}$ follows a unichan, we obtain from \cite[Theorem 8.4.5]{puterman2014markov} that for any $\pi_A$, we could obtain the optimal deterministic policy $\pi_S^*$ such that
	$\theta^*({\pi_A},\pi_S^*)\le\theta^*({\pi_A},\pi_S)$. Also, Algorithm 2 assures that for any $\pi_A$, $\theta^*({\pi^*_A},\pi_S^*)\le\theta^*({\pi_A},\pi_S^*)$. This leads to the conclusion that for any $\boldsymbol{\pi}_C=(\pi_S,\pi_A)\in \Upsilon$, we have that
	\begin{equation}
	\theta^*({\pi^*_A},\pi_S^*)\le\theta^*({\pi_A},\pi_S).
	\end{equation}
\end{proof}

\section{Simulation Results}\label{sectionV}
For the simulation setup, we set $\mathcal{A}_A=\left\{0,\cdots,10\right\}$, $\mathcal{S}=\left\{s_0,s_1,s_2\right\}$, $\mathcal{V}=\left\{v_0,v_1,v_2\right\}$ and the corresponding cost is:
\begin{equation}
\mathbf{C}_1(X_t,\Phi_t)=\begin{pmatrix}
0 &20 &50\\
0 &10 &20
\end{pmatrix},
\end{equation}
We assume $C_2(\pi_A(\hat{X}_t))$ and $C_3(\pi_A(\hat{X}_t))$ are both linear to the decision making with $C_2(\pi_A(\hat{X}_t)=C_g\cdot\pi_A(\hat{X}_t)$ and $C_3(\pi_A(\hat{X}_t))=C_I\cdot\pi_A(\hat{X}_t)$, where $C_g=8$ and $C_I=1$. 

\subsection{Comparing Benchmarks: Separate Design}\label{sectionVA}
For the decision making, we consider that the decision policy $\pi_A$ is predetermined by a greedy methodology:
\begin{equation}\label{greedy}
\begin{aligned}
\pi_A(\hat{X}_t)&=\mathop{\arg\min}\limits_{a_A\in\mathcal{S}_A}\mathop{\mathbb{E}}\limits_{\Phi_{t}}\left\{\left[C_1(\hat{X}_t,\Phi_t) - C_2(\pi_A(\hat{X}_t))\right]^+\right.
\\&\left.+ C_3(\pi_A(\hat{X}_t))\right\}.
\end{aligned}
\end{equation}
This greedy-based approach entails selecting the decision that minimizes the cost in the current step given that the estimate $\hat{X}_t$ is perfect. By calculating (\ref{greedy}), we obtain a greedy-based decision-making policy $\pi_A(\hat{X}_t) = [0,3,7]$. Under this decision-making policy, the following sampling benchmarks are considered for the sampling design:

\noindent $\bullet$ {\textbf{Uniform.}} Sampling is triggered periodically, \emph{i.e.}, $a_S(t)=\mathbbm{1}_{\left\{t=K*\Delta\right\}}$, where $K=0,1,2,\cdots$ and $\Delta\in\mathbb{N}^+$. For each $\Delta$, the sampling rate is calculated as $1/\Delta$ and the long-term average cost is obtained through Markov chain simulations.

\noindent $\bullet$ {\textbf{Age-aware.}} Sampling is executed when the AoI attains a predetermined threshold, \emph{i.e.},  $a_S(t)=\mathbbm{1}_{\left\{\mathrm{AoI(t)}>\delta\right\}}$, where the AoI-optimal threshold $\delta$ can be ascertained using the Bisection method delineated in Algorithm 1 of \cite{sun2019sampling}. In Fig. \ref{RatevsCostComparision}, we dynamically shift the threshold $\delta$ to explore the balance between sampling rate and \emph{utility}. 

\noindent $\bullet$ {\textbf{Change-aware}} Sampling is triggered whenever the source status changes, \emph{i.e.}, $a_S(t)=\mathbbm{1}_{\left\{X_t\ne X_{t-1}\right\}}$. The sampling frequency of this policy is heavily influenced by the system's dynamics: if the sources are transferred frequently, the sampling rate will be high, whereas if there are fewer transfers, the sampling rate will be low.

\begin{figure}
	\centering
	\includegraphics[angle=0,width=0.48\textwidth]{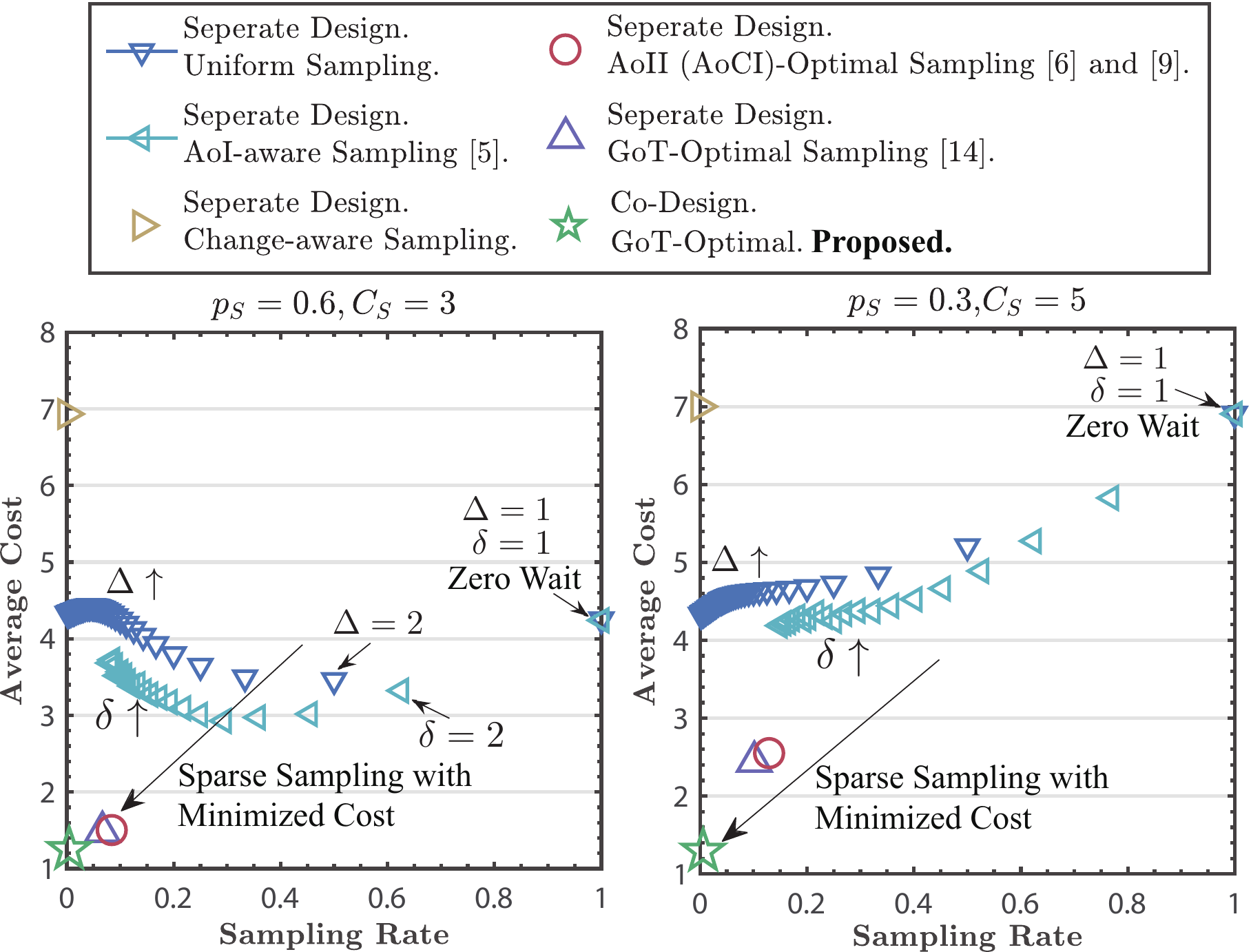}
	\caption{Average Cost vs. Sampling Rate under different policies and parameters setup. The series of Uniform and AoI-aware policies are obtained through shifting the intervals $\Delta$ and $\delta$. }\label{RatevsCostComparision}
\end{figure}

\noindent $\bullet$ {\textbf{Optimal AoII (also Optimal AoCI).}} From \cite{AoII}, it has been proven that the AoII-optimal sampling policy turns out to be $a_S(t)=\mathbbm{1}_{\left\{X_t\ne \hat{X}_t \right\}}$. From \cite{AoCI}, the AoCI-optimal sampling policy is $a_S(t)=\mathbbm{1}_{\left\{X_t\ne {X}_{t-\mathrm{AoI}(t)} \right\}}$. Note that $\hat{X}_t={X}_{t-\mathrm{AoI}(t)}$, these two sampling policies are equivalent. The sampling rate and average cost are obtained given this sampling policy and the greedy-based decision-making policy.

\subsection{Co-Design Through GoT}\label{sectionVB}

We notice that sampling and decision making are closely intertwined, highlighting the potential for further exploration of joint design. In this paper, we have introduced the RVI-Brute-Force-Search Algorithm (Algorithm \ref{Algorithm 2}) to distributively obtain the optimal joint policy. As shown in Fig. \ref{RatevsCostComparision}, the \emph{sampler-decision maker} co-design achieves the optimal long-term average \emph{utility} through only sparse sampling. Only information that carries crucial semantics for the decision making is sampled and transmitted, while others are filtered out. By incorporating a best-matching decision policy, the proposed goal-oriented, semantic-aware, and sparse sampling achieves superior performance compared to existing methods. In this regard, a goal-oriented, semantics-context-aware, sparse sampling is to achieve the maximized \emph{utility} through effective decision making.

\section{Conclusion}\label{sectionVI}
In this paper, we have investigated the GoT metric to directly describe the goal-oriented system decision-making\emph{utility}. Employing the GoT, we have formulated an infinite horizon Dec-POMDP problem to accomplish the integrated design of sampling and decision. To solve this problem, we have developed the RVI-Brute-Force-Search Algorithm to attain the optimal solution. Comparative analyses have substantiated that the proposed GoT-optimal \emph{sampler-decision maker} co-design can achieve sparse sampling and meanwhile maximize the \emph{utility}, signifying the realization for a sparse sampler and goal-oriented decision maker co-design.

\bibliographystyle{IEEEtran}
\bibliography{reference}

\end{document}